%% file: main_new.tex
\documentclass[letterpaper, 10 pt, conference]{ieeeconf}

\IEEEoverridecommandlockouts

\usepackage{amsmath,amsfonts,amssymb}
\usepackage{algorithm}
\usepackage{algpseudocode}
\usepackage{array}
\usepackage[caption=false,font=normalsize,labelfont=sf,textfont=sf]{subfig}
\usepackage{textcomp}
\usepackage{stfloats}
\usepackage{url}
\usepackage{verbatim}
\usepackage{graphicx}
\usepackage{mathtools}
\usepackage{xcolor}
\usepackage{amsthm}
\usepackage{booktabs}
\usepackage[]{hyperref} 
\usepackage{multicol}
\usepackage{comment}
\usepackage[font=footnotesize]{caption}

\usepackage[style=ieee,citestyle=numeric-comp,doi=false,isbn=false,url=false,eprint=false,maxnames=2,minnames=1]{biblatex}
\usepackage[left=0.67in, right=0.67in, top=0.833in, bottom=0.597in]{geometry}

\theoremstyle{definition}

\newtheorem{theorem}{Theorem}

\newtheorem{property}{Property}
\newtheorem{prop}{Proposition}

\DeclarePairedDelimiter\ceil{\lceil}{\rceil}
\DeclarePairedDelimiter\floor{\lfloor}{\rfloor}

\def\BibTeX{{\rm B\kern-.05em{\sc i\kern-.025em b}\kern-.08em
    T\kern-.1667em\lower.7ex\hbox{E}\kern-.125emX}}

\title{\LARGE \bf Linear Aggregate Model for Realizable Dispatch of Homogeneous Energy Storage}

\author{
Mazen Elsaadany$^{1,3}$,
Mads R. Almassalkhi$^{1,3}$, and
Simon~H.~Tindemans$^{2}$
\thanks{$^{1}$MRA and ME are with the Department of Electrical and Biomedical Engineering, University of Vermont, Burlington, VT, USA. ORCID (MRA): 0000-0003-3536-5563 ORCID (ME): 0009-0005-5756-4725.}
\thanks{$^{2}$SHT is with the Department of Electrical Sustainable Energy, Delft University of Technology, Mekelweg 4, 2628 CD, Delft, The Netherlands. ORCID: 0000-0001-8369-7568.}
\thanks{$^{3}$MRA and ME recognize support from DOE/PNNL via award DE-AC05-76RL01830. MRA was also supported by NSF Award ECCS-2047306.}
}
\begin{document}

\maketitle
\thispagestyle{empty}
\pagestyle{empty}

\begin{abstract}
 To optimize the dispatch of batteries, a model is required that can predict the state of energy (SOE) trajectory for a chosen open-loop power schedule to ensure admissibility (i.e., that schedule can be realized). However, battery dispatch optimization is inherently challenging when batteries cannot simultaneously charge and discharge, which begets a non-convex complementarity constraint. In this paper, we develop a novel composition of energy storage elements that can charge or discharge independently and provide a sufficient linear energy storage model of the composite battery. This permits convex optimization of the composite battery \textcolor{black}{dispatch} while ensuring the admissibility of the resulting (aggregated) power schedule and its disaggregation to the individual elements.
\end{abstract}


\input{Sections/Introduction}
\input{Sections/Composite_Battery_Model}
\input{Sections/Composite_Traectory_Dissaggregation}
\input{Sections/Simulation_Results}

\section{Conclusion}
A new linear dispatch model is presented and analyzed for a composite battery that permits unequal power sharing among its $N$ elements.
A priority-based disaggregation strategy is described and used to characterize sufficient conditions on the composite battery dispatch under which element constraints are guaranteed to be satisfied. Previous battery optimization implicitly assumed equal power sharing between all elements. By enabling unequal power sharing within a linear formulation, we have expanded the class of battery models that yield realizable power trajectories. {\color{black} Implementing the proposed controller is of interest and mainly requires communications to/from battery elements faster than $\delta t$ (similar to VPPs in the field today). 
Furthermore, embedding the presented linear models into energy system planning problems would inform the impact of battery systems more broadly.} {\color{black} Future work involves exploring other disaggregation strategies and non-linear battery models.}

\printbibliography
\end{document}

%% file: Sections/Introduction.tex
\section{Introduction}

Battery energy storage systems (BESS) are increasingly deployed to help facilitate the grid's uptake of variable renewable energy and improve reliability ~\cite{Battery_review_paper,BESS_simultaneous_participation,BESS_renewables}.
In many applications, BESS dispatch involves coordinating large numbers of energy storage elements. Examples include battery racks that comprise a single BESS asset, BESS container assets in a utility-scale battery storage facility as shown in Fig~\ref{fig:BESS_BlockDiagram}, or aggregated BESS units coordinated as a virtual power plant (VPP). 

Optimally scheduling the charging and discharging of storage elements requires solving \textcolor{black}{a non-convex optimization problem. Because most storage technologies do not permit simultaneous charging and discharging, \emph{complementarity constraints} are usually required to force the product of the charging and discharging powers to be equal to zero \emph{at each time step, and for each element}. 
This nonlinearity is often implemented using binary variables~\cite{MILP_for_CC,prat_2024}.} The non-convexity of the composite BESS dispatch can present a computational challenge when incorporating BESSs into more complex, large-scale power system optimization problems, such as optimal power flow (OPF) and economic dispatch~\cite{qi2025locationalenergystoragebid}.

\begin{figure}[t]
\centerline{\includegraphics[width=0.85\columnwidth]{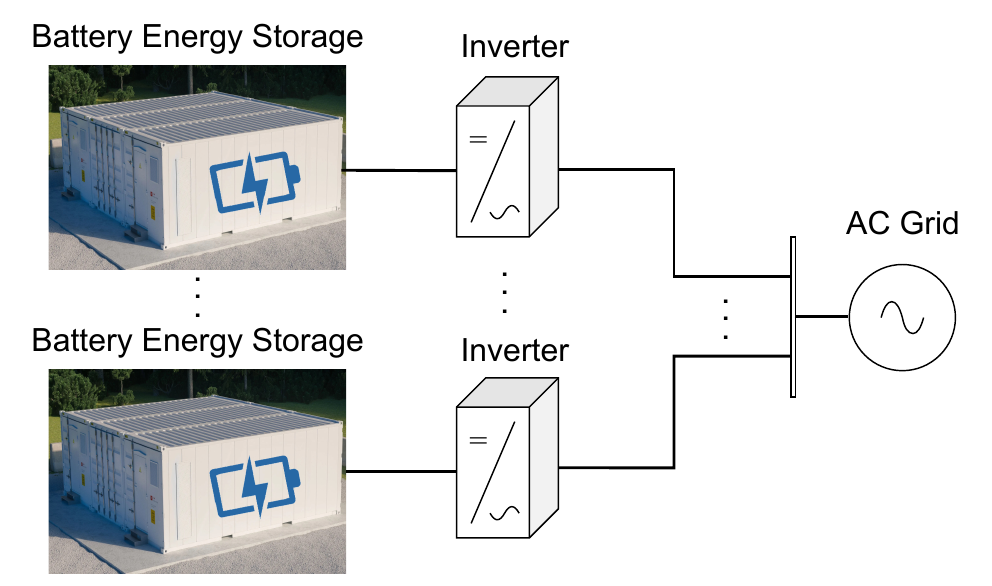}}
\caption{Example of a composite battery with AC-coupled battery elements.}
\label{fig:BESS_BlockDiagram}
\end{figure}

This motivates the need for a convex BESS model. A common approach is to obtain a convex relaxation by omitting the complementarity constraint, {\color{black} which results in a predicted state of energy (SOE) that is a provable lower bound on the realized SOE and can lead to unexpected SOE saturation and an inadmissible dispatch \cite{RobustNawaf, almassalkhi_model-predictive_2015}. Additionally, } the relaxation permits solutions that simultaneously charge and discharge \cite{prat_2024}. 

{\color{black} At first glance, such solutions may appear irrational because charging/discharging inefficiencies result in unnecessary dissipation of energy. Nevertheless, there are cases where such solutions are a rational outcome of a battery dispatch schedule, e.g., when energy prices are negative \cite{Baldick_Neg_Price_Bad}, or as a byproduct of minimum up- or downtime constraints in a centrally dispatched system}. {\color{black} Any dissipative effects on battery aging can be captured by adding a cycling cost to the objective function, so that the benefit must outweigh the cost.} Sufficient conditions under which simultaneously charging and discharging is optimal (i.e., rational) are given in~\cite{prat_2024}. 

{\color{black} Conversely, there are cases where simultaneous charging and discharging is guaranteed not to occur. } Sufficient \emph{ex-post} conditions for an exact relaxation in economic dispatch are provided in~\cite{Sufficient_Cond_for_Relax}. \emph{Ex-ante} conditions for exactness of the relaxed problem are presented in~\cite{Lin_Exact_Cond_Ex_Ante}, assuming that battery power limits, as well as initial and maximum energy capacities, satisfy certain conditions. Furthermore,~\cite{positive_prices,Jakob_Pos_Quadratic_Cost,prat_2024} show that for specific economic dispatch formulations with positive linear or quadratic cost functions, the relaxation at optimality would be exact with no simultaneous charging and discharging.

Various alternative problem formulations exist to reduce the reliance on the computationally demanding complementarity constraint. For a single time step, tighter relaxations have been proposed \cite{D.Pozo_paper, elgersma2024}, which can also heuristically be applied to time series, albeit without realizability guarantees. 
Authors in \cite{K.Baker_DCOPF_No_CC} avoid the use of the complementarity constraint when applied to DC Optimal Power Flow (DCOPF) but add a penalty term to the objective function to ensure the exactness of the relaxed model. 
Instead of modifying objective, the authors in ~\cite{RobustNawaf,Mazen_CDC_paper} present an alternative robust linear battery model which represents a convex restriction of the non-convex formulation that sacrifices optimality for guaranteed admissibility. The robust linear model omits the complementarity constraint but envelopes the actual SOE trajectories to guarantee that any feasible power trajectory is realizable. However, the envelopes beget conservativeness as a function of energy conversion efficiencies and the prediction horizon, which limits practicality.
\textcolor{black}{Now, a composite battery made up of identical storage elements can be scheduled with various methods. 
One way is to model and optimize the schedule of each storage element individually, requiring many complementarity constraints~\cite{equal_unequal}. A simpler and computationally efficient approach is to use a static \emph{equal power sharing} policy where all elements charge or discharge equally~\cite{equal_unequal,equal_unequal2}. The elements can then be aggregated into a composite battery whose dispatch is optimized, but still requires one non-convex complementarity constraint.} 
In summary, previous work in energy storage optimization that circumvents the complementarity condition either has been application-specific and holds only if specific criteria on the objective function and/or the battery parameters are met, or presents convex restrictions that may yield significantly sub-optimal solutions. Furthermore, previous work has generally (and often, implicitly) assumed an equal power-sharing policy, which limits how elements can be dispatched.
This manuscript revisits the modeling of aggregations of storage units by presenting a novel composite model of $N$ identical and independently controllable battery elements, where simultaneous charging and discharging power trajectories at the composite level can be implemented by inhomogeneous disaggregation to the individual elements. Throughout, BESS aggregations are used as an example, but the proposed work can be extended to different energy storage assets. 
The main contributions of the paper are as follows: 
\begin{enumerate}
    \item We present a novel linear model that provides a sufficient representation of aggregations of homogeneous energy storage elements. For large numbers of elements or and small control time steps, the representation asymptotically approaches the true feasibility envelope.
    \item We develop a priority-based strategy, based on element SOE, to perform composite-to-\textcolor{black}{elements} disaggregation.  
    \item {\color{black} We prove that all solutions of the linear aggregate model are realizable using the priority-based strategy, respecting element-wise power limits, energy limits, and complementarity constraints after disaggregation.  }
\end{enumerate}
Next, we develop the composite battery model in Section~II. Section~III presents the proposed disaggregation strategy and sufficient conditions for the admissibility of a composite power trajectory. Simulation results that demonstrate the performance of the proposed model and disaggregation strategy make up Section~IV. Conclusion and future work are in Section~V.

%% file: Sections/Composite_Battery_Model.tex
\section{Composite Battery Model}

{\color{black} In this section, we will present element-wise and composite models of BESS systems, composed of $N$ identical and independently controllable elements (e.g., as shown in Fig~\ref{fig:BESS_BlockDiagram}). }
For element $i\in \{1,\hdots, N\}$, the charge/discharge power is set in  \emph{controller time steps} $l$ with duration $\delta t$. The charge/discharge power between $l$ and $l+1$, the charge/discharge efficiency and the SOE at timestep $l$ are denoted by $P^i_\text{c/d}(l)$, $\eta_\text{c/d}$, and $E^i(l)$, respectively. A realizable BESS dispatch must satisfy power limits (${P}_\text{c/d,max}$), energy limits (${E}_\text{max}$), and complementarity constraints for all battery elements ($\forall i=1,\hdots,N$) and time steps ($\forall l=0,\hdots,L-1$). Note that all elements have identical power and energy limits. This results in \eqref{eqn:Elements Constraints}, where $E^i_0$ is the initial SOE of battery element $i$:
\begin{subequations}
\label{eqn:Elements Constraints}
    \begin{align}
        &E^i(l+1)=E^i(l)+ \delta t\eta_\text{c}P_\text{c}^i(l) -\delta t\frac{1}{\eta_\text{d}}P_\text{d}^i(l) \label{eqn:Elements SOC} \\
        &E^i(0) = E^i_0 \label{eqn:Elements SOC IC}\\
        &0\leq E^i(l+1)\leq {E}_\text{max}\label{eqn:Elements SOC Limit} \\
        &0\leq P_\text{c/d}^i(l)\leq {P}_\text{c/d,max}\label{eqn:Elements Pc Limit}\\
        &0=P_\text{c}^i(l)P_\text{d}^i(l) \label{eqn:Elements CC},
    \end{align}
    \end{subequations} 
{\color{black} where \eqref{eqn:Elements SOC Limit}-\eqref{eqn:Elements CC} represent energy limits, power limits and, complementarity constraints, respectively.} The scheduling of the composite battery is done with (optionally) longer \emph{scheduler time steps} of duration $\Delta t = M \delta t, M \in \mathbb{N}^+$. We denote these scheduling time steps with index $[k]$ in square brackets, where $k=0,\ldots, K-1$ and $L=MK$. Constant power in step $k$ implies the following constraint for all $l,k$ with $k =\floor{\frac{l}{M}}$: 
\begin{align}
    \sum_{i=1}^{N} P_\text{c}^i(l)=P_\text{c}[k]\quad  \land \quad 
    \sum_{i=1}^{N} P_\text{d}^i(l)=P_\text{d}[k]. 
    \label{eqn:Sum of Element}
\end{align}
Typically, it is assumed that a composite BESS system is modeled and operated under \emph{balanced} conditions: the dispatcher (e.g., EMS) 
maintains \textit{equal SOE} ($E_0^i = E_0~~\forall i$) and adopts an \textit{equal power sharing policy} across its elements.
This results in a composite battery with power and energy limits of $N{P}_\text{c/d,max}$ and $N{E}_\text{max}$ as shown in~\eqref{eqn:Composite Constraints}:
\begin{subequations}
    \label{eqn:Composite Constraints}
    \begin{align}
    E[k+1]&=E[k]+ \Delta t\eta_\text{c}P_\text{c}[k] -\Delta t\frac{1}{\eta_\text{d}}P_\text{d}[k] \label{eqn:Composite SOC} \\
    E[0] &= \sum_{i}E^i_0 \label{eqn:Composite SOC IC}\\
    0&\leq E[k+1]\leq N{E}_\text{max} \label{eqn:Composite SOC Limit}\\
    0&\leq P_\text{c/d}[k]\leq N{P}_\text{c/d,max}\label{eqn:Composite Pc Limit}\\
    0&=P_\text{c}[k]P_\text{d}[k]\label{eqn:Composite CC}.
\end{align}
\end{subequations}
Under equal power sharing, a complementarity constraint \eqref{eqn:Composite CC} at the composite battery level is sufficient to satisfy~\eqref{eqn:Elements CC} for all~$i$ and all~$l$. Under balanced operation, a composite dispatch satisfying~\eqref{eqn:Composite Constraints} would satisfy~\eqref{eqn:Elements Constraints} for all~$i$ and all~$l$.. 

The complementarity constraint is often neglected to obtain a convex relaxation of the problem. For example, as
\begin{subequations} \label{eqn:RelaxedTighter}
    \begin{align}
        \eqref{eqn:Composite SOC} &- \eqref{eqn:Composite Pc Limit} \label{eqn:Relaxed}\\
        \frac{P_\text{c}[k]}{N{P}_\text{c,max}}&+\frac{P_\text{d}[k]}{N{P}_\text{d,max}} \leq 1 \label{eqn:Relaxed Composite cutting plane},
    \end{align}
\end{subequations}
{\color{black}where the cutting plane in~\eqref{eqn:Relaxed Composite cutting plane} is one of the facets of the convex hull of the feasible set defined by \eqref{eqn:Composite Constraints}, typically used to tighten the relaxation. Even tighter relaxations have been developed \cite{D.Pozo_paper, elgersma2024} to reduce the prevalence of solutions with simultaneous charging and discharging in real-world problems. }
However, such relaxations can result in an optimal power schedule with simultaneous charging and discharging, which cannot be implemented under equal load sharing.

Such cases \emph{can} be addressed by a composite battery if we omit the equal load sharing requirement. 
This way, we can realize simultaneous charging and discharging dispatch schedules by appropriately disaggregating them 
across 
the $N$ elements while satisfying~\eqref{eqn:Elements CC}. This increases the capability of the composite battery at the cost of computational complexity with $N$ complementary constraints. 

In the remainder of this paper, we address this challenge by presenting a sufficient (but not unique) priority-based disaggregation strategy that engenders a \emph{linear} set of provably realizable composite power trajectories. These trajectories permit simultaneous charging and discharging of the BESS and approximate the full solution space arbitrarily closely.

%% file: Sections/Composite_Traectory_Dissaggregation.tex
\section{Composite Trajectory Disaggregation}

A control policy is used to disaggregate a desired piecewise constant composite power inputs $P_\text{c/d}[k]$ into $N$ individual element charge/discharge inputs, $\{P^i_\text{c/d}(l)\}_{i=1}^N$.  
 One such controller is the presented priority stack controller discussed below. 
 Sufficient conditions are developed that -- in combination with the controller -- enable the characterization of $P_\text{c}[k]$ and $P_\text{d}[k]$ that guarantees~\eqref{eqn:Elements Constraints} is satisfied.
\subsection{Priority Stack Controller (PSC)}

For every controller time-step $l$, the priority stack controller (PSC) sorts all $N$ elements based on the SOE. The element with the lowest SOE has the highest charging priority, while the element with the highest SOE has the highest discharging priority. \textcolor{black}{Note that similar control schemes have been proposed in~\cite{karan_kalsi} and investigated in a real-time operational setting in~\cite{equal_unequal,,equal_unequal2}}. 

The priority stack results in $N_\text{c}$ elements with lowest SOEs charging and $N_\text{d}$ elements with highest SOEs discharging, as follows:
    \begin{align}
        N_\text{c}(l) :=\ceil[\Big]{\frac{P_\text{c}[k]}{{P}_\text{c,max}}}
        \quad \land \quad 
        N_\text{d}(l) :=\ceil[\Big]{\frac{P_\text{d}[k]}{{P}_\text{d,max}}}.
        \label{eqn:NcNd_ceil eq}
    \end{align}
Disaggregating the composite charging trajectory for $P_\text{c}[k]>0$, the PSC sets $P^i_\text{c}(l)=P_\text{c,max}$ for the $N_\text{c}(l)-1 $ elements with the lowest SOE and the next-lowest element is assigned $P^i_\text{c}(l)=P_\text{c}[k]-(N_\text{c}(l)-1)P_\text{c,max}$. Discharging is implemented {\color{black} analogously.}
    The PSC is summarized in Algorithm~\ref{alg:Priority Stack} and begets Property~\ref{lemma:priority stack}.
    
    \begin{algorithm}
    \caption{Priority Stack Controller (PSC)}
    \label{alg:Priority Stack}
    \begin{algorithmic}[1]
    {\small 
        \State Given $P_\text{c}[k]$ and $P_\text{d}[k]$
        \For{each timestep $l$}
            \State Sort elements based on $E^i(l)$ in ascending order
            \State Calculate $N_\text{c/d}(l) := \ceil{\frac{P_\text{c/d}[k]}{P_\text{c/d,max}}}$
            \State Assign $P^i_\text{c/d}(l)$ for first $N_\text{c}(l)$/last $N_\text{d}(l)$ elements
            \State Update SOE for all elements: $E^i(l)\rightarrow E^i(l+1)$
        \EndFor
    }
    \end{algorithmic}
    \end{algorithm}
    
\begin{property}\label{lemma:priority stack}
Under PSC, if elements $i,j$ satisfy $E^i(l) \ge E^j(l)$,  then $P^i_\text{c}(l)\leq P^j_\text{c}(l)$ and $P^i_\text{d}(l)\geq P^j_\text{d}(l)$. 
\end{property}

\subsection{Guaranteeing Element Complementarity Constraints}
Next, we analyze the effects of the PSC on~\eqref{eqn:Elements Constraints}, on complementarity conditions~\eqref{eqn:Elements CC}, and on energy limits~\eqref{eqn:Elements SOC Limit}.
To ensure the complementarity constraints are satisfied at the element level, $(P_\text{c}^i(l)P_\text{d}^i(l)=0\quad \forall i,\forall l)$, the priority stack cannot issue a charging and discharging demand to the same battery element $i$ simultaneously.  
Thus, any overlap in the selection of $N_\text{c}(l)$ charging and $N_\text{d}(l)$ discharging elements imply a violation of the complementarity constraints, which must, therefore, satisfy
\begin{equation}\label{eqn:NcNd}
    N_\text{c}(l) + N_\text{d}(l)\leq N, \quad \forall l.
\end{equation}

\begin{prop} \label{prop: cutting plane}
Under the PSC, if $P_\text{c}[k]$ and $P_\text{d}[k]$ satisfy 
\begin{equation}
\label{eqn:new cut plane}
    \frac{P_\text{c}[k]}{{P}_\text{c,max}}+\frac{P_\text{d}[k]}{{P}_\text{d,max}} \leq N-1,
\end{equation}
then the resulting $N_\text{c}(l)$ and $N_\text{d}(l)$ satisfy~\eqref{eqn:NcNd}.
\end{prop}

\begin{proof}
Define $\Gamma_\text{c}:=\lfloor \frac{P_\text{c}[k]}{{P}_\text{c,max}}\rfloor$ and $\Delta_\text{c}:= \frac{P_\text{c}[k]}{{P}_\text{c,max}} - \Gamma_\text{c} \in [0,1)$ {\color{black}to be the integer and fractional components of the quantity $\frac{P_\text{c}[k]}{{P}_\text{c,max}}$. }
Then $\ceil[\Big]{\frac{P_\text{c}[k]}{{P}_\text{c,max}}} \le \Gamma_c + 1$ and, from \eqref{eqn:new cut plane}, we get 
\begin{align}\label{eqn:ceil second term}
    \ceil[\Big]{\frac{P_\text{d}[k]}{{P}_\text{d,max}}} \leq \ceil[\Big]{N-1 -(\Gamma_\text{c} + \Delta_\text{c}) } = N -1-\Gamma_\text{c}.
\end{align}
Adding the inequalities, we get 
$N_c(l)+N_d(l)\le N$ $\forall l$. 
\end{proof}
Thus, adding~\eqref{eqn:new cut plane} to the composite formulation in~\eqref{eqn:Composite Constraints} will satisfy element complementarity constraints under PSC. 
The feasible regions dictated by~\eqref{eqn:NcNd} and \eqref{eqn:new cut plane} are shown in Fig.~\ref{fig:feasible region}. 
\begin{figure}[ht]
\centering
\includegraphics[width=.475\columnwidth]{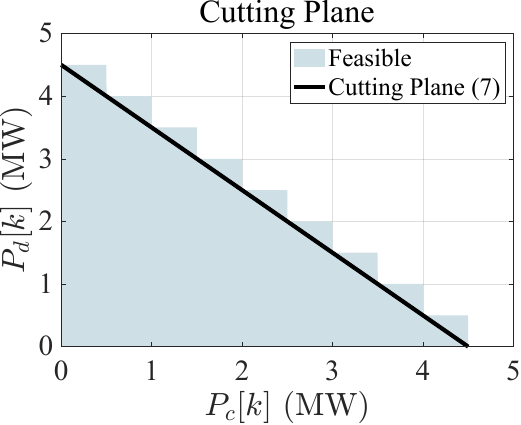}~
\includegraphics[width=.475\columnwidth]{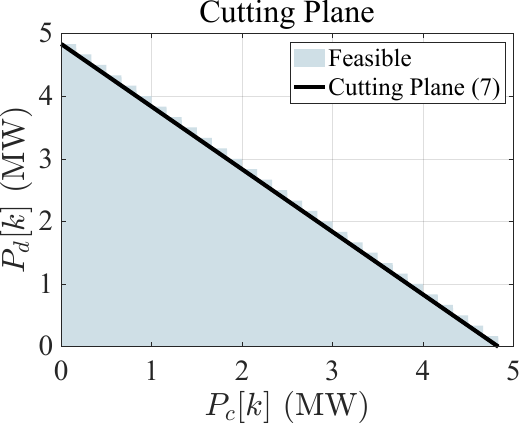}
\caption{Feasible region for $N=10$ and $P_\text{c/d,max}=0.5$MW (left) and $N=30$ and $P_\text{c/d,max}=0.17$MW (right). SOE limits are omitted here.}
\label{fig:feasible region}
\end{figure}
\subsection{Relating Element and Composite Energy Constraints}
Similarly, we want to relate element energy limits to those in the composite formulation. Under PSC, the battery elements are charged/discharged individually every $l$, which results in differing SOE trajectories. The spread of element SOEs can impact the realizability of the composite dispatch when individual elements hit their SOE limits. 
Constraint violations can be avoided by quantifying the maximum possible SOE spread under PSC and using that as an energy buffer applied to the composite battery energy limits. 

\begin{prop} \label{prop:SOC difference boud}
Under the PSC, a composite dispatch that satisfies 
\begin{equation}\label{eqn:Composite_Energy_Limits}
    N\varepsilon\leq E[k]\leq N (E_\text{max}-\varepsilon),
\end{equation}
with
\begin{equation} \label{eq:DEmax-def}
    \varepsilon := \delta t\left(\eta_\text{c}{P}_\text{c,max} + \frac{{P}_\text{d,max}}{\eta_\text{d}}\right),
\end{equation}
will not violate any individual element's energy limit, provided that ($i$) $\max_{i, j}\{ |E^i_0-E^j_0|\} \leq \varepsilon$ and ($ii$) the control time step $\delta t=\Delta t/M$ is sufficiently small to ensure $\varepsilon \le \frac12 E_\text{max}$.

\end{prop}

\begin{proof}
Define the maximum difference in SOE across all elements as $\Delta E(l):= \max_{i, j}\{ |E^i(l)-E^j(l)|\}$ (i.e., the SOE spread).   
Consider element SOEs $E^i(l)$ and $E^j(l)$ and, without loss of generality, let $E^i(l)\ge E^j(l)$. Using \eqref{eqn:Elements SOC}  the following is obtained,
\begin{align}
      E^i(l+1) - E^j(l+1) &= \left(E^i(l) - E^j(l)\right) +\\ \frac{\Delta t}{M}\eta_\text{c}\left(P^i_\text{c}(l)-P^j_\text{c}(l) \right)& - \frac{\Delta t}{M\eta_\text{d}} \left(P^i_\text{d}(l)-P^j_\text{d}(l) \right). \notag
\end{align}
From Property~\ref{lemma:priority stack} and \eqref{eqn:Elements Pc Limit}, we have $-P_\text{c,max}\le P^i_\text{c}(l)-P^j_\text{c}(l)\le 0$ and $0\le P^i_\text{d}(l)-P^j_\text{d}(l) \le P_\text{d,max}$, which implies 
\begin{subequations}\label{eqn:SOC_diff_bound}
    \begin{align}
        E^i(l+1) - E^j(l+1) \leq & E^i(l) - E^j(l) \leq \Delta E(l), \\
        E^i(l+1) - E^j(l+1) \ge & -\frac{\Delta t}{M}\left(\eta_\text{c}P_\text{c,max} + \frac{P_\text{d,max}}{\eta_\text{d}}\right) = -\varepsilon, \label{eqLastIneq}
    \end{align}
\end{subequations}
where the inequality in~\eqref{eqLastIneq} uses $E^i(l)- E^j(l) \ge 0$.
 Considering all $i,j$, we can bound $\Delta E(l+1) \leq  \max\left\{ \Delta E(l), \varepsilon \right\}$.
 Moreover, we see that  $\Delta E(l) \leq \varepsilon \Rightarrow \Delta E(l+1) \leq \varepsilon$, so it follows that for a composite battery with initial SOE spread $\Delta E(0)\leq \varepsilon$ (condition (i)), we maintain $ \Delta E(l) \leq \varepsilon, \forall l$.

 Because energy linearly increases/decreases within each scheduling time step $k$ under PSC, \eqref{eqn:Composite_Energy_Limits} implies its equivalent for controller time steps $l$: $N\varepsilon\leq E(l)\leq N (E_\text{max}-\varepsilon)$. Dividing by $N$ begets $\varepsilon\leq \overline{E}(l)\leq E_\text{max}-\varepsilon$, where $\overline{E}(l):=\sum_j E^j(l)/N$ is the average element SOE. Because the SOE spread $\Delta E(l) \le \varepsilon$, we must have $\overline{E}(l) - \varepsilon \le E^i(l) \le \overline{E}(l) + \varepsilon, \forall i$. 
Thus, \eqref{eqn:Elements SOC Limit} holds for all $i$ and all $l$.

 Lastly, we require condition (ii) ($\varepsilon \le \frac12 E_\text{max}$) to ensure that \eqref{eqn:Composite_Energy_Limits} does not result in an empty set (infeasibility). This can is achieved with a sufficiently large value of $M$.
\end{proof}

\subsection{Realizable Composite Power Trajectories}

The following theorem provides a linear formulation for realizable dispatch of a composite battery.
\begin{theorem} \label{the:RCB}
If $\max_{i, j}\{ |E^i_0-E^j_0|\}  \le \varepsilon$, then any feasible composite battery dispatch schedule $\{P_\text{c}[k],P_\text{d}[k]\}_{k=0}^{K-1}$ satisfying 
\begin{subequations}
    \label{eqn:Realizable Composite}
    \begin{align}
        \eqref{eqn:Composite SOC}, 
          & && k=0,\hdots, K-1
        \label{eq:relaxedModelhidden} \\
        \eqref{eqn:Composite SOC IC}, 
          & && \\
        \frac{P_\text{c}[k]}{N{P}_\text{c,max}} + \frac{P_\text{d}[k]}{N{P}_\text{d,max}} 
        \leq \frac{N-1}{N}, &&& k=0,\hdots, K-1
        \label{eqn:RCB cut plane} \\
        N\varepsilon \leq E[k] \leq N\left(E_\text{max}-\varepsilon\right), & 
         && k=0,\hdots, K
        \label{eqn:RCB energy buffer} \\
        P_\text{c}[k], P_\text{d}[k] \geq 0, 
        & &&  k=0,\hdots, K-1 \label{eq:RCB power limit}
    \end{align}
\end{subequations}
will satisfy~\eqref{eqn:Elements Constraints} for all $N$ elements under PSC and is, thus, realizable.    
\end{theorem}
\begin{proof}
Eqs.~\eqref{eqn:Elements SOC} and \eqref{eqn:Sum of Element} imply \eqref{eqn:Composite SOC} and \eqref{eqn:Elements SOC IC} implies \eqref{eqn:Composite SOC IC}.      From Proposition~\ref{prop: cutting plane}, \eqref{eqn:RCB cut plane} ensures satisfaction of~\eqref{eqn:Elements CC}. Similarly, with 
$\max_{i, j}\{ |E^i_0-E^j_0|\}  \le \varepsilon$,
Proposition~\ref{prop:SOC difference boud} guarantees that~\eqref{eqn:Elements SOC Limit} holds, and the PSC with \eqref{eq:RCB power limit} and \eqref{eqn:Sum of Element} ensure that \eqref{eqn:Elements Pc Limit} holds for all $N$ elements. Thus, any feasible composite dispatch is realizable.
\end{proof} 

Note that for a fixed composite battery, if $N$ increases,
the cutting plane in~\eqref{eqn:RCB cut plane} approaches the cutting plane in~\eqref{eqn:Relaxed Composite cutting plane}. Furthermore, for $M\gg 1$, $\varepsilon$ in \eqref{eq:DEmax-def} tends to zero. That is, with faster control and more granular elements,~\eqref{eqn:Realizable Composite} approaches the relaxed model in~\eqref{eqn:RelaxedTighter}; however, with guaranteed realizability. 

The energy buffer $\varepsilon$ in~\eqref{eqn:RCB energy buffer} restricts the set of feasible energy trajectories relative to~\eqref{eqn:Composite Constraints} and~\eqref{eqn:Relaxed}. The feasible energy trajectories in the robust linear model presented in~\cite{RobustNawaf} are also a subset of the feasible energy trajectories in~\eqref{eqn:Composite Constraints} and~\eqref{eqn:Relaxed}. However, the proposed model in~\eqref{eqn:Realizable Composite} has a constant energy buffer $\varepsilon$ that is independent of the length of time horizon considered. This is in contrast to the robust model which has an energy buffer that increases with the length of the time horizon which, in the worst case, significantly reduces the feasible set of energy trajectories and can yield sub-optimal solutions. This is illustrated in Fig.~\ref{fig:feasible_energies} where the (worst case) feasible set of energy trajectories for the robust model and the RCB model for $\Delta t=1\text{hr}$ and $M=1$ ($\delta t=1\text{hr}$) and $M=5$ ($\delta t = 12\text{min}$) are compared for the battery parameters of a Tesla Powerwall shown in Table~\ref{tab:Battery_Parameters}. Clearly, smaller control time steps result in a smaller $\varepsilon$ and thus a larger feasible set of energy trajectories for the proposed model, while the robust model remains unchanged.

Finally, we note that the PSC policy is a greedy sufficient policy that fully activates BESS elements in sequence. More elaborate policies (greedy or not) can be developed that preferentially spread element loading for most aggregate dispatch signals, without violating the conditions of Theorem~\ref{the:RCB}. 

\begin{table}[h]
    \caption{Individual BESS Element Parameters}
    \begin{center}
    \setlength\tabcolsep{5pt}
     \centering
     \label{tab:Battery_Parameters}
    \begin{tabular}{lrl}
    \toprule
    \textbf{Parameters}        & \textbf{Value} & \textbf{Unit}\\ \midrule
    Charge/discharge efficiencies, $\eta_\text{c}/\eta_\text{d}$         & 0.95          & -\\ 
    Maximum charge/discharge power, $P_\text{max}$      & 5.00              & kW \\ 
    Maximum energy capacity, $ E_{\text{max}}$ & 13.5 & kWh\\\bottomrule
    \end{tabular}
    \end{center}
    \end{table}

    \begin{figure}[h]
        \centerline{\includegraphics[width=0.85\columnwidth]{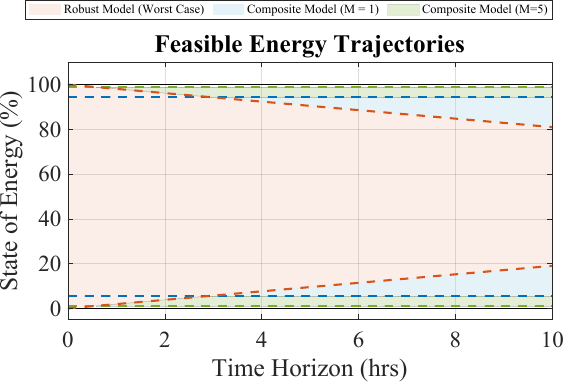}}
        \caption{Feasible set energy trajectories for the relaxed model, the robust model and the proposed realizable composite battery model (RCB).}
        \label{fig:feasible_energies}
    \end{figure}

%% file: Sections/Simulation_Results.tex
\section{Simulation Results}

We apply the realizable composite battery (RCB) model in~\eqref{eqn:Realizable Composite} to two relevant use cases and compare its performance against the tighter relaxed model given in~\eqref{eqn:RelaxedTighter} and the robust model in~\cite{RobustNawaf} used under balanced operation. Also, a Mixed-Integer Linear Program (MILP) formulation with both equal and unequal power-sharing for $N$ elements in~\eqref{eqn:Elements Constraints} used for comparison. The battery element parameters given in Table~\ref{tab:Battery_Parameters} were used. The two use-cases are a) power reference tracking and b) revenue maximization. {\color{black}All cases are modeled using JuMP in Julia and were solved using Gurobi~v12.0.0. All cases were run on a Macbook Pro with an M1 Pro processor. }

\subsection{Power Reference Tracking}
In this use case, the goal is to minimize the power tracking error for a given reference signal as shown below,
\begin{equation}
   \min_{P_\text{c}[k],P_\text{d}[k]} \qquad \sum_{k=0}^{K-1} \left((P_\text{c}[k]-P_\text{d}[k])-P_{\text{ref}}[k]\right)^2.
\end{equation}
A timestep of $\Delta t = 3$ min, along with $N = 100$ was used. The RCB model was used three times with $M = 1$, $M = 5$, and $M = 10$ respectively.  The results from using each model are shown in Fig.~\ref{fig:sim_res} and the solve times along with the mean squared error (MSE) are shown in Table~\ref{tab:Power_Ref}. Note that the relaxed model resulted in a dispatch with simultaneous charging and discharging, which cannot be realized under equal power sharing \textcolor{black}{and results in the discrepancy between predicted MSE and actual MSE obtained while realizing the optimal dispatch}. The net power trajectory was implemented which caused the BESS to saturate causing the discrepancy between the predicted and realized MSE values.

\begin{figure}[h]
\centering
\includegraphics[width=0.85\columnwidth]{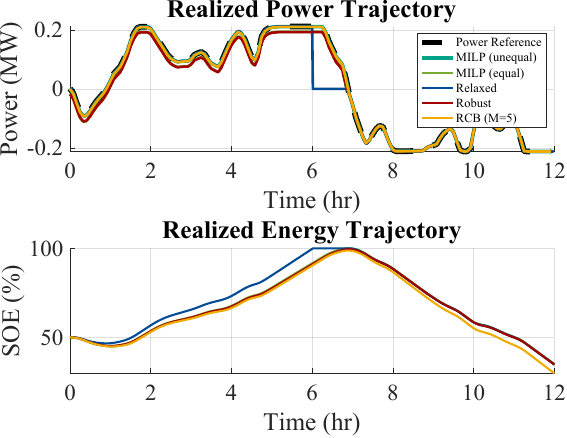}
\caption{Open-loop implementation of power optimal reference tracking power schedules using MILP, relaxed, robust and proposed RCB models. }
\label{fig:sim_res}

\end{figure}

\begin{table}[htbp]
    \caption{Power Reference Tracking Performance Metrics ($N=100$)}
    \begin{center}
    \setlength\tabcolsep{5pt}
     \centering
     \label{tab:Power_Ref}
    \begin{tabular}{lrrr}
        \toprule
        {\textbf{Model}} 
         & \textbf{Solver}& \textbf{Predicted}& \textbf{Actual} \\
          {}& \textbf{Time }&\textbf{MSE (kW$^2$)}& \textbf{MSE (kW$^2$)} \\
          \midrule
    MILP (unequal) & 492.60 s & 0 & 0  \\ 
    MILP (equal) & 34.37 s & 200.0& 200.0  \\ 
    \midrule
    Relaxed & 45 ms & 0 & 1889 \\ 
    Robust & 54 ms & 212 & 212 \\ 
    \midrule
    RCB (M = 1) & 50 ms & 23.93 & 23.93 \\ 
    RCB (M = 5) & 50 ms & 0.231 & 0.231\\ 
    RCB (M = 10) & 50 ms & 0 & 0\\ 

      \bottomrule
    \end{tabular}
    \end{center}
    \end{table}

{\color{black}The MILP model under unequal power sharing was able to track the reference signal with no tracking error. The MILP under equal power sharing couldn't track the reference signal as well, despite both solving to 0\%  optimality gap, indicating that there is a benefit of simultaneously charging and discharging in this use case. The MILP, under both equal and unequal power sharing, came with a significant computational burden compared to the rest of the formulations, attributed to the extra binary variables added to enforce the complementarity constraint.} The other formulations are linear and, as a result, have significantly faster solve times. The relaxed model predicts zero tracking error; however, the solution was unrealizable and caused the battery to saturate around $t = 6$ hrs, leading to a large realized tracking error. The robust model yields a realizable yet sub-optimal solution. The RCB model shows improved tracking error with increasing values of $M$, but even surpassing the robust model with $M=1$.

\subsection{Revenue Maximization}
In this use case, the goal is to maximize the predicted revenue from the composite battery system given a forecasted price signal. The objective function is given by,
\begin{equation}
    \max_{P_\text{c}[k],P_\text{d}[k]} \qquad \sum _{k=0}^{K} C[k](P_\text{d}[k]-P_\text{c}[k]),
    \label{eqn:Revenue}
\end{equation}
where $C[k]$ is the price at time $k$. The price signal was obtained from the California Independent System Operator (CAISO) day-ahead market prices, 
plotted in Fig~\ref{fig:price_forecast}.

\begin{figure}[h]
    \centering
    \includegraphics[width=7cm]{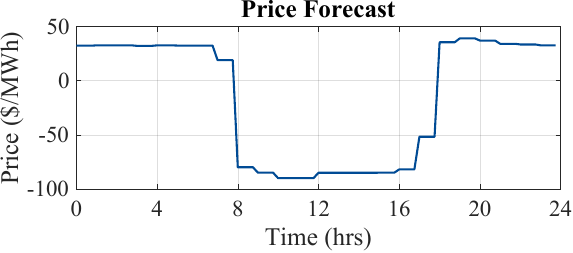}
    \caption{CAISO day-ahead price forecasts}
    \label{fig:price_forecast}
\end{figure}
A timestep of $\Delta t = 15$ min, along with $N = 100$ was used. The RCB model was tested for $M =$ 1, 5, and 10 which correspond to $\delta t = $ 15 min, 3 min and 1.5 min respectively. An additional case was added with $M=900$ ($\delta t = 1s$) to gain insight on model performance as $M \rightarrow \infty. $The results for the different models are shown in Table~\ref{tab:Revenuse}. 
    \begin{table}[h]
        \caption{Revenue Maximization Performance metrics ($N=100$)}
        \begin{center}
        \setlength\tabcolsep{5pt}
         \centering
         \label{tab:Revenuse}
        \begin{tabular}{lrrrr}
        \toprule
        {\textbf{Model}} 
         & \textbf{Solver}& \textbf{Predicted}& \textbf{Actual}& \textbf{Optimality} \\
        {}& \textbf{Time }&\textbf{Revenue (\$)}& \textbf{Revenue (\$)}&\textbf{Gap (\%)} \\
          \hline
        MILP (unequal) & $90$ mins & 2088.99 & 2088.99 & 0.02 \\ 
        MILP (equal) & 23 mins & 2088.97 & 2088.97   & 0 \\ 
        \midrule
        Relaxed & 1.5 ms & 2092.80 & 1940.35 & 0 \\ 
        Robust & 2.2 ms  & 1866.42 & 1866.42 & 0 \\ 
        \midrule
        RCB (M = 1)  & 1.7 ms & 1642.24 & 1642.24& 0 \\ 
        RCB (M = 5)  & 1.9 ms & 2000.50 & 2000.50& 0 \\ 
        RCB (M = 10) & 1.7 ms & 2045.24 & 2045.24& 0 \\ 
    RCB (M = 900) & 1.6 ms & 2089.48 & 2089.48 & 0\\
          \bottomrule
        \end{tabular}
        \end{center}
        \end{table}

{\color{black}The MILP models produced high revenues that could also be realized, but at substantial computational cost, with unequal power sharing taking significantly longer to solve and was not able to solve to 0\% optimality gap. } The relaxed model predicts the highest revenue but underestimates the SOE, leading to unpredicted battery saturation and a realized performance that is significantly worse. The robust model yields a realizable yet suboptimal solution. The RCB model for small values of $M$ results in a sub-optimal solution which can be attributed to the larger timestep used in this use-case leading to a larger value or $\varepsilon$ and, in turn, a more restricted feasible set. {\color{black}However, for $M\ge5$ ($\delta t \le 3$ min), the RCB model surpasses all other linear models and for larger values of $M$ the predicted revenue steadily increases, maintaining realizability. 

We note that for $M=900$ the revenue even surpasses the element-wise MILP solution, because individual elements are able to switch from charging to discharging and back at a faster controller timestep when implemented by the PSC. A MILP solution with \textcolor{black}{the control time step $\delta t$ would again achieve or exceed this result, but at much larger computational cost. The RCB model captures what is feasible at the fast control time step $\delta t$ without added complexity. }